\documentclass[11pt]{article}
\usepackage[utf8]{inputenc}
\usepackage[left=1in, top=1in, right=1in, bottom=1in]{geometry}
\usepackage{soul}
\usepackage{amsthm}
\usepackage{amsmath}
\usepackage{amssymb}
\usepackage{xcolor}
\usepackage{setspace}
\usepackage{booktabs} 
\usepackage{algorithm}
\usepackage{algorithmic}
\usepackage[colorlinks=true, citecolor=blue, linkcolor=red, urlcolor=black]{hyperref}
\usepackage[switch]{lineno}
\usepackage[nocomma, short]{optidef}

\usepackage{amsfonts}
\usepackage{mathtools}
\usepackage{bbm}
\usepackage{natbib}
\usepackage{mdframed}
\mdfsetup{
 skipabove=4pt,
 skipbelow=4pt,
 leftmargin=0pt,
 rightmargin=0pt,
 innertopmargin=3pt,
 innerbottommargin=3pt,
 innerleftmargin=3pt,
 innerrightmargin=3pt,
 frametitleaboveskip=3pt,
 frametitlebelowskip=3pt,
 frametitlerule=true,
 splittopskip=2\topsep}
\usepackage{enumitem}

\newtheorem{definition}{Definition}
\newtheorem{theorem}{Theorem}
\newtheorem{lemma}{Lemma}

\newtheorem{example}{Example}
\newtheorem*{remark}{Remark}
\newtheorem{fact}{Fact}

\DeclareMathOperator*{\E}{\mathbf{E}}
\newcommand{\dd}{\mathrm{d}}

\ttfamily

\begin{document}
\title{Revenue Maximization Mechanisms for an Uninformed Mediator with Communication Abilities}

\author{
Zhikang Fan \\ Renmin University of China\\ \texttt{fanzhikang@ruc.edu.cn}
\and
Weiran Shen\\ Renmin University of China \\ \texttt{shenweiran@ruc.edu.cn}
}

\maketitle

\begin{abstract}
Consider a market where a seller owns an item for sale and a buyer wants to purchase it. Each player has private information, known as their type. It can be costly and difficult for the players to reach an agreement through direct communication. However, with a mediator as a trusted third party, both players can communicate privately with the mediator without worrying about leaking too much or too little information. The mediator can design and commit to a multi-round communication protocol for both players, in which they update their beliefs about the other player's type. The mediator cannot force the players to trade but can influence their behaviors by sending messages to them.

We study the problem of designing revenue-maximizing mechanisms for the mediator. We show that the mediator can, without loss of generality, focus on a set of direct and incentive-compatible mechanisms. We then formulate this problem as a mathematical program and provide an optimal solution in closed form under a regularity condition. Our mechanism is simple and has a threshold structure. Additionally, we extend our results to general cases by utilizing a variant version of the ironing technique. In the end, we discuss some interesting properties revealed from the optimal mechanism, such as, in the optimal mechanism, the mediator may even lose money in some cases.
\end{abstract}
\setcounter{tocdepth}{1} 

\onehalfspacing

\newpage 

\section{Introduction}
\label{sec:introduction}
Consider a market where a seller wishes to sell an item to a buyer. The item's quality is only known to the seller, while the buyer's valuation of the item depends on both the item's quality and the buyer's type. Both parties are interested in knowing the other player's private information but may be unwilling to disclose too much about their own. On one hand, revealing too much information gives the other player an informational advantage. If the seller knows the buyer's type, they can set a price accordingly to extract the full surplus as their revenue. Similarly, if the buyer knows the quality of the item, they can certainly utilize the information to make a better purchase decision. On the other hand, revealing no information may hinder the occurrence of a trade. For example, a buyer may deem the price unworthy if they have information about the item's quality. Therefore, this lack of information makes it difficult for both players to come to a trade agreement on their own. Even after the purchase decision is made, the players also need to deal with extensive paperwork, which imposes high costs on both of them. 

This problem is ubiquitous in real-world applications and has led to the emergence of mediators between the two sides. For instance, in real estate markets, a broker or a realtor often acts as a mediator between a seller and a buyer. In fact, according to the National Association of Realtors~\citep{realtor}, in 2021, 87\% of buyers purchased their homes, and 90\% of sellers sold their homes through an agent or a broker in the US. Such a mediator typically benefits both sides and facilitates trades that may not happen on their own.

Apart from real estate markets, mediators also play a crucial role in numerous other applications: online ad exchanges match publishers with advertisers based on the features of both sides; ride-sharing platforms connect passengers with drivers based on their geographical information; lodging services match tourists with rooms based on their needs.

The aforementioned observations prompt us to investigate how a mediator can maximize their revenue by charging fees for providing such a service. We consider a model where the seller and the buyer can interact with a mediator via communication. Both players have private types that are drawn from publicly known distributions. The mediator has no access to the players' private types, but can privately communicate with them based on a specific communication protocol that may involve multiple rounds. Players update their beliefs about the other player's type after each round. In the end, the mediator decides whether to recommend a trade or not, and both players must decide whether to follow the recommendation. We assume there is no outside option for both players since there may be a significantly high cost involved (from negotiation, paperwork, etc.) if they choose to trade on their own.

However, unlike most standard mechanism design problems, the mediator cannot force the players to follow the trade recommendation. As a result, the mediator needs to carefully design the communication protocol and the pricing strategy to incentivize both players to follow their recommendations while maximizing revenue. 

In this work, we take the perspective of the mediator and explore the design of a communication protocol that enables the mediator to elicit and reveal information. We also aim to determine how the mediator can price their service to maximize the expected revenue.

\subsection{Our Contributions}
We formally define the mechanism space for this problem. Any mechanism in our mechanism space induces an extensive-form game with incomplete information once the mediator commits to a specific protocol, where each player has a prior belief about the other player's private information. Thus, we use the Perfect Bayesian equilibrium as our solution concept and prove a variant of the celebrated revelation principle, tailored for our setting. We show that the mediator can, without loss of generality, focus on the set of direct and incentive-compatible mechanisms. 

Based on the above results, we formulate this problem as a mathematical program and provide an optimal solution in closed form, subject to a regularity condition. Our optimal mechanism falls into the category of \emph{threshold mechanisms}. For those irregular cases, we adopt a variant of the ironing procedure proposed by Myerson \cite{myerson1981optimal} and prove that the ironed functions give the same revenue with respect to the original functions. Lastly, we discuss some interesting properties of the optimal mechanism. For example, a somewhat counterintuitive finding is that sellers with higher-quality items are less likely to make a sale. Additionally, in some cases, the mediator may incur a loss.

\subsection{Additional Related Works}
\paragraph{Information selling.}
Closest to our paper in the literature of information selling is the work by~\citet{Liu2021OptimalPO}. They consider sequential interactions between an information seller and an information buyer and study the optimal pricing problem. However, our paper differs from theirs in that we consider a setting where a mediator, who has no information advantage, interacts with two privately informed players. \citet{babaioff2012optimal} consider a similar setting and provide conditions under which the information seller can achieve optimal revenue by using a one-round mechanism. \citet{chen2020selling} consider the same setting as~\cite{babaioff2012optimal}, but with additional budget constraints. 

\paragraph{Bilateral trade.} Another relevant but significantly different direction is bilateral trade. The seminal work~\cite{myerson1983efficient} characterizes the set of mechanisms for a profit-maximizing monopolistic broker in bilateral trade. However, in their setting, both players cannot opt out of the mechanism after participation, and thus it can be modeled as a single-round mechanism. In contrast, in our setting, players need to decide whether to stay in the mechanism or follow the mediator's recommendations. Therefore, the revelation principle based on Bayesian Nash equilibrium does not apply to our setting. There is also a line of work focusing on improving efficiency through interaction~\cite{alon2015welfare,dobzinski2014economic}. More recently,~\citet{Mao2022InteractiveCI} consider a setting where a buyer and a seller can establish a communication protocol beforehand, and both of them have commitment power. Following this protocol, the players send verifiable signals to each other in turn. However, in our setting, only the mediator has commitment power, and each player privately communicates with the mediator to update their beliefs.

\paragraph{Information design.} Our work is also related to the literature on information design, commonly referred to as ``Bayesian persuasion''~\cite{kamenica2011bayesian}. In a standard Bayesian persuasion model, an informed sender who has no ability to take actions aims to persuade an uninformed receiver to take actions in their favor by way of signaling. Recently, there have also been studies on its variants, e.g., an informed sender persuades an informed receiver~\cite{kolotilin2017persuasion}; an informed sender persuades a set of uninformed receivers~\cite{Castiglioni2021SignalingIB,xu2020tractability}. In contrast, our problem lies in the area where an information sender with no informational advantage privately sends signals to two informed receivers. There is also a body of work that studies Bayesian persuasion with mediators~\cite{Shen2018ACC,Hennigs2019ConflictPB,Mahzoon2022HierarchicalBP,Arieli2022BayesianPW}. However, the role of the mediators in these settings is significantly different from ours.

\paragraph{Communication equilibrium.} Another relevant concept is communication equilibrium~\cite{forges1993five}, where the uninformed mediator can recommend actions to the players based on the information obtained from them. We apply this notion to a real-world application. For a more detailed discussion on the connection between Bayesian persuasion and incomplete information correlated equilibrium, we refer interested readers to~\cite{Bergemann2016InformationDB}.

\section{Preliminaries}
\label{sec:preliminary}

Consider a market with two agents: a \emph{seller} $s$ and a \emph{buyer} $b$. The seller has an item for sale, while the buyer seeks to purchase it. Let $t\in T$ be the private type of the buyer. Denote by $q\in Q$ the quality of the item, which can be viewed as the private type of the seller. We assume that $t$ and $q$ are random variables independently drawn from publicly known distributions $F(t)$ and $G(q)$, respectively. Additionally, we assume that both $F(t)$ and $G(q)$ are differentiable, with probability density functions $f(t)$ and $g(q)$ that are supported on $[t_1,t_2]$ and $[q_1, q_2]$, respectively.

Let $v(t,q)$ denote the valuation of a buyer of type $t$ obtaining an item of quality $q$. We assume that the function $v(t,q)$ is monotone non-decreasing in the buyer's type $t$ for any $q\in Q$. In particular, we assume that $v(t, q)$ is linear in $t$ and takes the form $v(t, q) = \alpha_1(q)t + \alpha_2(q)$, with $\alpha_1(q)>0$ for all $q\in Q$. The seller has a reserve price $r(q)$ for the item, which is assumed to be proportional to the quality of the item, i.e., $r(q)=kq$ for some constant $k\ge 0$. Alternatively, $r(q)$ can be interpreted as the seller's own valuation of the item.

Now suppose that there exists a third party (i.e., the \emph{mediator}) who has a private communication channel through which they can communicate with either the buyer or the seller privately, collecting information from the players or revealing information to them. In the end, the mediator will either recommend ``trade'' with a suggested price or ``not trade'' at all. The mediator can also charge them for providing the service. 

Note that the revenue received by the mediator comes solely from charging fees to the agents and is not related to the actual trade price. However, we can simplify our analysis by assuming that the buyer pays the trade price to the mediator, who then pays the seller. With this modification, the utilities of all three parties remain the same, and all monetary transfers occur only between the mediator and the two agents. Throughout the paper, we make the assumption that the seller and the buyer cannot trade with each other without the mediator, as the cost of reaching a trade agreement and handling all the paperwork is formidably high. Therefore, whenever the mediator recommends ``not trade'', the agents cannot trade on their own, even if the trade would benefit both of them. 

\section{Mechanism Space}
\label{sec:space}
We investigate the problem of designing revenue-maximizing mechanisms for the mediator. In this section, we will describe the set of mechanisms considered in this paper. Therefore, the optimal mechanism is the one that yields the highest expected revenue among the set of mechanisms.

\begin{definition}[General communication protocol]
\label{def:general_mechanism}
A general communication protocol is a mechanism that induces a finite extensive-form game between the agents. The game can be described by a game tree, where each non-leaf node $h$ is one of the following types:
\begin{itemize}
    \item {\bfseries buyer node}, where the buyer chooses an action;
    \item {\bfseries seller node}, where the seller chooses an action;
    \item {\bfseries signaling-buyer node}, where the mediator privately sends a message to the buyer. Any such node $h$ has a prescription of the mediator's behavior, which is a probability distribution $\psi_h$ over the set of its children $C(h)$;
    \item {\bfseries signaling-seller node}, which is analogous to the signaling-buyer node but has prescribed behavior $\psi_h$;
    \item {\bfseries charging-buyer node}, which has only one child node and is associated with a monetary transfer (possibly negative) $t_h$ from the buyer to the mediator;
    \item {\bfseries paying-seller node}, which is similar to the above node except that the transfer is from the mediator to the seller;
    \item {\bfseries recommendation node}, where the mediator decides whether to recommend the agents to trade. We allow for randomized recommendations. Therefore, each recommendation node has two children (one for ``trade'' and the other for ``not trade''), and is associated with the mediator's behavior, which is a probability distribution over the two children;
    \item {\bfseries buyer decision node}, where the buyer decides whether to trade when receiving ``trade'' recommendation. Once the buyer refuses to trade, the game ends and goes to a leaf (terminal) node;
    \item {\bfseries seller decision node}, which is analogous to the above node except that the decision is made by the seller.
\end{itemize}
We require that the path from the root to any terminal node include one and only one recommendation node. Also, we require at least one buyer decision node and one seller decision node after the child node ``trade'' of any recommendation node to guarantee the players' right to reject. Except at the recommendation node, any interaction between the buyer and the mediator cannot be observed by the seller, and vice versa. 
\end{definition}

Our mechanism space generalizes a similar definition called the ``generic interactive protocol'' proposed by~\citet{babaioff2012optimal}. Similar to other mechanism design problems, our mechanism proceeds as follows:
\begin{enumerate}
    \item The mediator announces the mechanism to the agents;
    \item The agents play the game described by the game tree.
\end{enumerate}

In our definition, each node involving the mediator is associated with a specific behavior. Therefore, the induced game tree is only well-defined once both the seller and the buyer are aware of the mediator's intended behaviors. We make the standard assumption that the mediator has commitment power and will behave exactly as committed.

Note that the buyer node and the seller node can be regarded as nodes where the two agents send messages to the mediator since their actions rely on their private information. In fact, in a direct mechanism (which will be defined later), the action sets for the agents are precisely their type sets.

\subsection{Game Re-Formulation and Solution Concept}
Since we assume that the mediator has commitment power, the induced game is actually played between the seller and the buyer. Thus, given a mechanism, we can re-formulate the game and view the mediator's actions as chance moves. 

Formally, given a mechanism $M$, a corresponding two-player extensive-form game can be defined as a tuple $\Gamma = \langle  N, T, H, Z, \mathcal{H}, A, \mathcal{I}, u \rangle$, where

\begin{itemize}
    \item $N = \{b, s\}$ is the set of players;
    \item $T$ is a game tree, obtained by adding a \emph{nature} node to the tree described by $M$, where nature chooses the types of players at the beginning of the game;
    \item $H$ is a set of non-terminal nodes;
    \item $Z$ is a set of terminal nodes;
    \item $\mathcal{H} = \{H_0, H_s, H_b\}$ is a partition of $H$, where $H_0$ is the set of chance move nodes, and $H_s$ and $H_b$ are the sets of nodes where the seller and the buyer moves, respectively;
    \item $A$ is a function that maps a non-terminal node $h\in H$ to a set of available actions $A(h)$;
    \item  $\mathcal{I}= \{\mathcal{I}_i\}_{i\in N}$ is a collection of information partitions, where each $\mathcal{I}_i$ is a partition of $H_i$. An element $I \in \mathcal{I}_i$ is called an \emph{information set} of player $i$. Every information set $I$ satisfies $A(h) = A(h'), \forall h, h' \in I$. Thus for any information set $I\in \mathcal{I}_i$, we denote by $A(I)$ the set $A(h)$ for all $h \in I$;
    \item $u = (u_i)_{i\in N}$ is a collection of payoff functions, where each $u_i:Z \mapsto \mathbb{R}$ is the payoff function of player $i$.
\end{itemize}
In the game tree $T$, any node can be uniquely determined by the path from the root to the node. This path is also called the history of the node. So with a slight abuse of notation, we use $h$ or $z$ to denote both a node and its history.

We assume quasi-linear utility for both the buyer and the seller. Thus, we have 
\begin{align*}
    u_b(z)=
    \begin{cases}
    v(t, q)-\tau_b(z) & \text{if $z$ leads to a trade}\\
    -\tau_b(z) & \text{otherwise}
    \end{cases},
\end{align*}
where $\tau_b(z)$ is the total payments made by the buyer along the path from the root to node $z$, and $t$ and $q$ can be inferred from $z$, as $z$ is a history including chance moves. Note that trade occurs if and only if the mediator recommends ``trade'' and both agents choose to follow the recommendation. Similarly, for the seller, we have:
\begin{align*}
    u_s(z)=
    \begin{cases}
    \tau_s(z) & \text{if $z$ leads to a trade}\\
    \tau_s(z) + r(q) & \text{otherwise}
    \end{cases},
\end{align*}
where $\tau_s(z)$ is the total payments from the mediator to the seller.

A \emph{strategy} $\beta_i(I)$ of player $i$ is a function that assigns a probability distribution over $A(I)$ to each information set $I\in \mathcal{I}_i$. Let $B_i$ denote the set of all possible strategies of player $i$. Given a strategy profile $\beta = (\beta_i)_{i\in N}$, one can easily calculate the probability $\rho(z|\beta)$ of reaching a terminal node $z$ using the structure of the game tree.

A \emph{belief} is a function $\mu: H \mapsto [0,1]$ that assigns to each node $h \in H$ a probability such that the probabilities of the node in any information set sum up to 1, i.e., $\sum_{h\in I}\mu(h) = 1, \forall I \in \mathcal{I}_i$, for all players $i\in N$. The belief $\mu(h)$ can be interpreted as the probability that the player believes they are at node $h$ when $I$ is reached.

An \emph{assessment} is a combination of a strategy profile and a belief $(\beta, \mu)$. Given an information set $I$ of player $i$, let $Z_I$ be the set of all terminal nodes that can be reached from some nodes in $I$. Given an assessment $(\beta, \mu)$ and information set $I\in \mathcal{I}_i$, the conditional probability of reaching a terminal node $z\in Z_I$ is denoted as $\rho(z|\beta,\mu, I)$. Then, the conditional expected payoff of player $i$ for $(\beta, \mu)$ at information set $I$ is defined as $U_{i,I}(\beta|\mu) = \sum_{z\in Z_I} \rho(z|\beta, \mu, I)u_i(z)$.

Now we are ready to define the solution concept used in this paper.

\begin{definition}[Perfect Bayesian Equilibrium, PBE]
An assessment $(\beta^*, \mu^*)$ is a perfect Bayesian equilibrium of $\Gamma$ if the following conditions hold for each player $i\in N$.
\begin{itemize}
    \item \emph{Sequential rationality}. For every information set $I \in \mathcal{I}_i$ and every $\beta_i \in B_i$, 
    \begin{gather*}
        U_{i,I}(\beta_i^*, \beta^*_{-i}|\mu^*) \ge U_{i,I}(\beta_i,\beta^*_{-i}|\mu^*);
    \end{gather*}
    \item \emph{Consistency}. For every information set $I \in \mathcal{I}_i$ and every $h \in I$, 
    \begin{gather*}
        \mu^*(h) = \frac{\rho(h|\beta^*)}{\sum_{h'\in I}\rho(h'|\beta^*) },
    \end{gather*}
    where $\rho(h'|\beta^*)$ is the probability of reaching node $h'\in I$ given $\beta^*$.
\end{itemize}
\end{definition}

\begin{remark}
Note that the above belief update equation is only well-defined if $\sum_{h'\in I} \rho(h'|\beta^*) > 0$. Otherwise, $\mu^*(h)$ can be arbitrary, as the probability of reaching the information set is 0.

Simply put, the players are sequentially rational if they always maximize their utilities based on their beliefs, and consistency requires that the players update their beliefs using the Bayes rule.
\end{remark}

\subsection{Direct Mechanisms and the Revelation Principle}
In this section, we recast the celebrated revelation principle in our setting, showing that the mediator can narrow their focus to the set of direct and incentive-compatible mechanisms without loss of generality.

Before defining direct mechanisms, we first introduce signaling schemes, which formalize how the mediator reveals information. Given a signal set $\Sigma$, a signaling scheme $\pi: T \times Q \mapsto \Delta(\Sigma)$ is a mapping from the players' type profile to a distribution over the signal set $\Sigma$. This way of revealing information is also called ``Bayesian persuasion'' or an ``experiment'' in some research~\cite{kamenica2011bayesian}.

Upon receiving signal $\sigma$, the buyer with type $t$ updates their belief about the seller's type $q$ using the standard Bayes rule:
\begin{align*}
    g(q|\sigma, t) &= \frac{\pi(\sigma|t, q)g(q)f(t)}{\int_{q'\in Q}\pi(\sigma|t, q')g(q')f(t) \,\mathrm{d} q'} \\
    &= \frac{\pi(\sigma|t, q)g(q)}{\int_{q'\in Q}\pi(\sigma|t, q')g(q') \,\mathrm{d} q'}.
\end{align*}
Similarly, the seller can also obtain a posterior belief:
\begin{align*}
    f(t|\sigma, q) = \frac{\pi(\sigma|t, q)f(t)}{\int_{t'\in T}\pi(\sigma|t',q)f(t') \,\mathrm{d} t'}.
\end{align*}

In this paper, we use signal schemes to describe how the mediator sends recommendations. The signal set $\Sigma$ has only two signals\footnote{One might think that the reason for this is that we only need to consider the so-called ``responsive experiments'' as shown in \cite{bergemann2018design}. However, this is not true, since in our paper, the receivers' available actions depend on the actual signal sent by the sender. Thus their result does not apply here.}, i.e., $\Sigma=\{0, 1\}$, where 0 corresponds to ``not trade'' while 1 corresponds to ``trade''. This is because we assume that the two agents cannot trade without the mediator. Thus the mediator should clearly specify whether or not they would recommend them to trade.

We define the payment functions for the two players as follows: $P_b: T \mapsto \mathbb{R}$ is the payment made by the buyer of type $t$ to the mediator, and $P_s: Q \mapsto \mathbb{R}$ is the payment from the mediator to the seller with quality $q$. 

Now we are ready to define direct mechanisms.
\begin{definition}[Direct Mechanism]
A direct mechanism is described by a tuple $(\pi, P_b, P_s)$, and proceeds as follows:
\begin{enumerate}
    \item The mediator announces $\pi$, $P_b$ and $P_s$;
    \item The buyer and seller privately report (possibly untruthfully) their types $t$ and $q$ to the mediator;
    \item The mediator decides whether to recommend the agents to trade according to $\pi(t, q)$;
    \item Upon receiving the signal, the players update their belief, and decide whether to follow the recommendation;
    \item If the trade occurs, the mediator charges the buyer $P_b(t)$ and pays the seller $P_s(q)$.
\end{enumerate}
\end{definition}
\begin{remark}
In the above definition, although the mediator already knows both $t$ and $q$ when calculating the payments, the payment functions $P_b(t)$ and $P_s(q)$ still only depend on their respective types $t$ and $q$. This is because the dependency on both $t$ and $q$ would reveal additional information.
\end{remark}

\begin{definition}[Incentive Compatibility]
A direct mechanism is incentive compatible if, for both players, reporting their true types and following the mediator's recommendation form a PBE of the game induced by the mechanism.
\end{definition}

One can easily verify that a direct mechanism can also be represented as a general communication protocol. We now show that the mediator can, without loss of generality, focus solely on direct and incentive compatible mechanisms.

\begin{theorem}
\label{theorem:revelation}
For any general mechanism $M$, there exists a direct, incentive compatible mechanism that achieves the same expected revenue as in any PBE of the game induced by $M$.
\end{theorem}

\section{Problem Analysis}
In this section, we analyze the problem and formulate it as a mathematical program. For simplicity, we use $\pi(t, q)$ to denote the probability of sending signal 1, when the reported type profile is $(t, q)$. 

The mediator's goal is to design $\pi(t, q)$, $P_b(t)$ and $P_s(q)$ to maximize their expected revenue:
\begin{gather*}
	\int_{q\in Q}\int_{t\in T} \pi(t, q)[P_b(t) - P_s(q)]f(t)g(q) \,\dd t \dd q.
\end{gather*}
We also need to pose some constraints to ensure that the agents are willing to participate and report truthfully.

\paragraph{Individual rationality (IR).} 
The utility of a buyer who chooses not to participate is 0. Therefore, the expected utility of a buyer with type $t$ who reports truthfully and follows the recommendation is:
\begin{align}
		U_b(t) &=  \E_{q\sim G}[\pi(t, q)(v(t, q) - P_b(t)) ] \nonumber\\
    	&= \int_{q\in Q} \pi(t, q)[v(t, q) - P_b(t)]g(q) \,\mathrm{d} q.\label{eq:utility buyer}
\end{align}
Note that the seller values the item $r(q)$ even if the trade does not occur. Thus the expected utility of a seller with type $q$ reporting truthfully and following the recommendation is:
\begin{align}
		U_s(q) &=  \E_{t\sim F}[\pi(t, q)P_s(q) ] + \E_{t\sim F}[(1-\pi(t, q))r(q)] \nonumber\\
    	&= \int_{t\in T} \pi(t, q)[P_s(q) - r(q)]f(t) \,\mathrm{d} t +r(q).\label{eq:utility seller}
\end{align}
Therefore, to ensure the buyer's participation, we require:
\begin{align}
	\int_{q\in Q} \pi(t, q)[v(t, q) - P_b(t)]g(q) \,\mathrm{d}q \ge 0. \label{eq:IR buyer}
\end{align}
Similarly, for the seller, we have:
\begin{align}
\label{eq:IR seller}
	\int_{t\in T} \pi(t, q)[P_s(q) - r(q)]f(t) \,\mathrm{d} t +r(q) \ge r(q) .
\end{align}

For convenience, we define the notion of ``seller surplus'', which expresses the seller's additional utility gain from participating in the mechanism, as follows:
\begin{align}
	SU(q) = \int_{t\in T}\pi(t, q)[P_s(q) - r(q)]f(t) \,\mathrm{d}t.
\end{align}
Therefore, the IR constraint for the seller is equivalent to non-negative seller surplus. 

\paragraph{Incentive compatibility (IC).} To satisfy the IC constraint, we need two steps. The first step is to ensure that following the mediator's recommendation is the best option for both agents, assuming that they truthfully reported their types in previous steps. We refer to this property as the \emph{obedience} property. The second step is to ensure that both agents obtain a higher utility from truthfully reporting their types than misreporting any other types, even if they choose not to follow the recommendation.

The first step requires that for each agent, choosing to trade maximizes the utilities when signal 1 (the ``trade'' signal) is received. For the buyer, we have:
\begin{align*}
    \int_{q\in Q} g(q|1, t)  v(t, q) \,\dd q - P_b(t) \ge 0.
\end{align*}
With some simple algebraic manipulations, we obtain:
\begin{align}
\label{eq:obedience-buyer}
     \int_{q\in Q} \pi(t, q)[v(t, q) - P_b(t)]g(q) \,\mathrm{d}q \ge 0.
\end{align}
Similarly, for the seller, we have:
\begin{align}
\label{eq:obedience-seller}
  \int_{t\in T} \pi(t,q)[P_s(q) - r(q)] f(t) \,\mathrm{d}t \ge 0.
\end{align}
Note that when the mediator sends signal 0 (the ``not trade'' signal), we do not need to guarantee obedience, as we assume that the agents cannot trade without the mediator. 

It is interesting to observe that the above obedience constraints turn out to be exactly the same as the IR constraints. Thus we can safely ignore the obedience constraints.

For the second step, we first consider the constraints for the buyer. We still only need to consider the case when signal 1 is received, as the agents cannot trade without the mediator. Upon receiving signal 1, the expected utility of trade for a buyer with type $t$ misreporting $t'$ is given by:
\begin{gather}
	U_b(t';t) = \int_{q \in Q} \pi(t',q)[v(t,q) -P_b(t') ]g(q) \,\mathrm{d}q.
\end{gather}
Note that for a buyer of type $t$, $\pi(t', q)$ may not be obedient anymore, i.e., after misreporting their type, the buyer may choose not to follow the trade recommendation if doing so leads to a negative utility. So the expected utility for the buyer is $\max \{U_b(t';t), 0\}$. Thus the incentive compatibility constraints for the buyer become the following:
\begin{gather}
	U_b(t) \ge \text{max} \{U_b(t';t), 0\}.
\end{gather}
Interestingly, $U_b(t) \ge 0$ is already implied by the IR constraint (Constraint \eqref{eq:IR buyer}), so the only useful constraint is 
\begin{gather}
\label{eq:IC buyer}
	U_b(t) \ge  U_b(t';t).
\end{gather}

For the seller, upon receiving signal 1, the expected utility of trade for a seller with type $q$ misreporting $q'$ is 
\begin{gather}
	U_s(q';q) = \int_{t \in T} \pi(t, q')[P_s(q') - r(q)]f(t) \,\mathrm{d} t + r(q),
\end{gather}
where $\pi(t, q')$ may also not be obedient for the seller. With arguments similar to that for the buyer, we obtain that the IC constraint for the seller is the following:
\begin{gather}
\label{eq:IC seller}
    U_s(q) \ge U_s(q';q).
\end{gather}
The seller's IC constraint can also be equivalent to misreporting does not result in a higher seller surplus, that is:
\begin{align}
	SU(q) \ge SU(q';q),
\end{align}
where $SU(q';q) = \int_{t\in T} \pi(t, q') [P_s(q') - r(q)]f(t) \,\mathrm{d} t$. 

Combining all the constraints together, we can formulate the mechanism design problem as the following mathematical program:
\begin{maxi}
{}   
{\int_{q\in Q} \int_{t\in T} \pi(t, q)[P_b(t) - P_s(q)]f(t)g(q) \,\mathrm{d}t \mathrm{d}q}  
{\label{eq:LP}}  
{}  
\addConstraint{U_b(t)}{\ge 0, \quad}{\forall t \in T}
\addConstraint{SU(q)}{\ge 0, \quad}{\forall q \in Q}
\addConstraint{U_b(t)}{\ge U_b(t';t), \quad}{\forall t, t' \in T}
\addConstraint{SU(q)}{\ge SU(q';q), \quad}{\forall q, q' \in Q}
\end{maxi}

\section{The Optimal Mechanism}
\label{sec:optimal}
In this section, we apply a variant of the Myersonian approach and give an optimal solution in closed form to the optimization problem \eqref{eq:LP}. 

Before we present our main result, we need the following definition.
\begin{definition}[Virtual Value Function and Virtual Cost Function]
For any random variable $w$ with PDF $x(w)$ and CDF $X(w)$, the virtual value function $\phi^-(w)$ and the virtual cost function $\phi^+(w)$ are defined as:
\begin{gather*}
    \phi^-(w)=w-\frac{1-X(w)}{x(w)}, \quad \phi^+(w)=w+\frac{X(w)}{x(w)}.
\end{gather*}
\end{definition}
The definition of the virtual value function is the same as in \cite{myerson1981optimal}, and the virtual cost function is also standard in the literature (see, e.g., \cite{manelli1995optimal,myerson1983efficient}).

Let $\phi^-_b(t)$ be the buyer's virtual value function and $\phi^+_s(q)$ the seller's virtual cost function. Our main result is built upon the following regularity condition:
\begin{definition}[Regularity]
	A problem instance is regular if both functions $\phi^-_b(t)$ and $\frac{k \phi^+_s(q) - \alpha_2(q)}{\alpha_1(q)}$ are monotone non-decreasing.
\end{definition}

This regularity condition is also standard in the literature \cite{myerson1981optimal,cai2011extreme,krishna2001convex} and is satisfied for a wide range of distributions. Our solution belongs to the following category of \textit{threshold mechanisms} when the problem instance is \emph{regular}.

\begin{definition}
A mechanism $(\pi, P_b, P_s)$ is called a threshold mechanism if there exist monotone functions $\lambda(t)$ and $\eta(q)$, such that the mediator recommends ``trade'' as long as $\lambda(t)\ge \eta(q)$, i.e.,
\begin{gather*}
    \pi(t, q) = \begin{cases}
    1 & \text{if } \ \lambda(t) \ge  \eta(q) \\
    0 & \text{otherwise}
    \end{cases}.
\end{gather*}
In this case, $\pi(t, q)$ is fully described by function $\lambda(t)$ and $\eta(q)$.
\end{definition}

Now we are ready to present our optimal mechanism.
\begin{theorem}
Suppose that a problem instance satisfies the regularity condition. Then the threshold mechanism with the following threshold functions and payment functions is an optimal mechanism:
\begin{gather*}
    \lambda(t) = \phi^-_b(t)\quad \text{and}\quad \eta(q) = \frac{k \phi^+_s(q) - \alpha_2(q)}{\alpha_1(q)},
\end{gather*}
\begin{align}
    P_b^*(t) = &\frac{1}{\int_{q\in Q} \pi^*(t, q)g(q)\,\mathrm{d}q} \left[ \int_{q\in Q}\pi^*(t, q) v(t,q) g(q) \,\mathrm{d}q \right.\nonumber \\
    &\left. - \int_{t_1}^t \int_{q\in Q}\alpha_1(q)\pi^*(x, q)g(q) \,\mathrm{d}q \mathrm{d}x \right],\label{eq:opt_buyer_payment}\\
    P_s^*(q) = &\frac{1}{\int_{t\in T}\pi^*(t, q)f(t) \,\mathrm{d}t} \left[\int_{t\in T}\pi^*(t, q) r(q) f(t) \,\mathrm{d}t \right. \nonumber\\
    &\left. + k\int_{q}^{q_2} \int_{t\in T}\pi^*(t, x)f(t) \,\mathrm{d}t \mathrm{d}x  \right]\label{eq:opt_seller_payment}.
\end{align}
\label{thm:optimal_mechanism}
\end{theorem}

\begin{example}\label{example:1}
Consider an example where both the buyer's type and the seller's quality are uniformly distributed in the interval $[1,2]$, i.e., $T=Q=[1,2]$, and $f(t)=g(q)=1, \forall t, q$. We set $\alpha_1(q)=q$ and $\alpha_2(q)=0$. Let $r(q) = 1.5q$, i.e., $k=1.5$.

In this example, we have $\phi^-_b(t) = t - \frac{1-F(t)}{f(t)} = 2t-2, \phi^+_s(q) = q+\frac{G(q)}{g(q)} = 2q-1$. Therefore $\lambda(t) = \phi^-_b(t) = 2t-2$ is non-decreasing in $t$ and $\eta(q)= k\frac{\phi^+_s(q)}{q} = 3 - \frac{1.5}{q}$ is also non-decreasing in $q$, satisfying the regularity condition. In the optimal mechanism, the mediator will recommend them to trade if $\lambda(t)\ge \eta(q)$, or equivalently, when $t\ge 2.5 - \frac{1.5}{2q}$. For the buyer, there are two cases.
\begin{itemize}
    \item When $t < 1.75$, we have $t < 2.5-\frac{1.5}{2q},\forall q$. This means that the mechanism will never recommend trade, or equivalently, $\pi^*(t, q) = 0, \forall q$. The payment is $P_b^*(t) = 0$. For these buyers, there is no trade the mediator charges nothing.
    \item When $t \ge 1.75$, the mechanism will recommend trade when $q \le \frac{1.5}{5-2t}$. 
    In this case, the buyer's payment function becomes:
    \begin{gather*}
        P_b ^*(t) =  0.5+\frac{0.9375}{2.5-t}.
    \end{gather*}
\end{itemize}

For the seller, there are also two cases.
\begin{itemize}
    \item When $q > 1.5$, then $2.5 - \frac{1.5}{2q} > 2$. This means the mediator will never recommend trade since $t$ is at most 2. So we have $\pi^*(t, q) = 0$ and $P_s^*(q) =0$  for all $t$ in this case.
    \item When $q \le 1.5$, we have that $\pi^*(t, q) = 1$ when $t \ge 2.5 - \frac{1.5}{2q}$. If $q < 1.5$, the payment function for the seller $P_s^*(q)$ becomes:
    \begin{gather*}
        P_s^*(q) =\frac{9q}{6-4q} \ln\left(\frac{1.5}{q}\right).
    \end{gather*}
    And if $q=1.5$, we can take the limit and get $P_s^*(q)=2.25$.
\end{itemize}
\end{example}

The rest of the section is devoted to the proof of Theorem \ref{thm:optimal_mechanism}. We first introduce the following quantities:
\begin{gather}
    R^{\pi}_b(t) = \int_{q\in Q}\alpha_1(q) \pi(t, q) g(q)  \,\mathrm{d}q \label{eq:weighted buyer},\\
    R^{\pi}_s(q) = \int_{t \in T} \pi(t, q) f(t) \,\mathrm{d}t. \label{eq:weighted seller}
\end{gather}

Note that $R^{\pi}_b(t)$ and $R^{\pi}_s(q)$ can be viewed as the expected probability of being recommended to trade, except that $R^{\pi}_b(t)$ is weighted by $\alpha_1(q)$.

We call a mechanism $(\pi, P_b, P_s)$ \emph{feasible} if it satisfies the constraints in Program \eqref{eq:LP}. The following lemma gives a characterization of feasible mechanisms.
\begin{lemma}
\label{lem:property}
A mechanism $(\pi, P_b, P_s)$ is feasible if and only if it satisfies the following constraints:
\begin{align}
    & R^{\pi}_b(t) \text{ is monotone non-decreasing in $t$}. \label{eq:buyer_monotone}\\
    & R^{\pi}_s(q) \text{ is monotone non-increasing in $q$}. \label{eq:seller_monotone}\\
    & U_b(t) = U_b(t_1) + \int_{t_1}^t R^{\pi}_b(x) \,\mathrm{d}x \label{eq:IC property buyer} \\
    & SU(q) = SU(q_1) - k\int_{q_1}^q R^{\pi}_s(x) \,\mathrm{d}x \label{eq:IC property seller}\\
    & U_b(t_1) \ge 0 \label{eq:IR property buyer}\\
    & SU(q_2) \ge 0 \label{eq:IR property seller}
\end{align}
\end{lemma}

Before we derive the optimal mechanism, we first need to rewrite the revenue of the mediator.
\begin{lemma}
\label{lem:revenue}
	The mediator's expected revenue of any feasible mechanism $(\pi, P_b, P_s)$ can be written as:
\begin{align}
    Rev(\pi,P_b,P_s)= & \int_{q_1}^{q_2} \int_{t_1}^{t_2}\pi(t, q)[\alpha_1(q) \phi^-_b(t) +\alpha_2(q)- k \phi^+_s(q)] \nonumber\\
    & \cdot f(t)g(q)\,\mathrm{d}t\mathrm{d}q-U_b(t_1) - SU(q_2) .\label{eq:prove revenue}
\end{align}
\end{lemma}
 
The following result shows that the mechanism described in Theorem \ref{thm:optimal_mechanism} is feasible. 
\begin{lemma}
For a regular problem instance, the mechanism $(\pi^*, P_b^*, P_s^*)$ defined in Theorem \ref{thm:optimal_mechanism} is feasible.
\label{lem:feasible}
\end{lemma}
We only provide a proof sketch of Theorem \ref{thm:optimal_mechanism} due to space limit.
\begin{proof}[Proof sketch of Theorem \ref{thm:optimal_mechanism}]
    According to Lemma \ref{lem:feasible}, our mechanism is feasible. The definition of $\pi^*$ implies that our mechanism maximizes the first term in Equation \eqref{eq:prove revenue}. And using the definition of $P^*_b(t)$ and $P^*_s(q)$, we can show that $U_b(t_1)=0$ and $SU(q_2)=0$. This means that our mechanism optimizes the 3 terms in Equation \eqref{eq:prove revenue} simultaneously, and hence optimal. 
\end{proof}

We also find that the optimal mechanism must satisfy that $U_b(t_1) = U_s(q_2) = 0$.
\begin{lemma}
\label{lem: payment}
    A feasible mechanism $(\pi, P_b, P_s)$ is a revenue-maximizing mechanism, then it must satisfy $U_b(t_1) = 0, U_s(q_2) = 0$.
\end{lemma}

\section{The Optimal Mechanism in General Cases}
The optimal mechanism provided in section \ref{sec:optimal}  relies crucially on the regularity condition, i.e., the monotonicity of $\phi_b^-(t)$ and $\frac{k\phi_s^+(q) - \alpha_2(q)}{\alpha_1(q)}$. Without this condition, the optimal mechanism would not even be feasible since it violates either \eqref{eq:buyer_monotone} or \eqref{eq:seller_monotone}. In this section, we extend our results to irregular cases. For ease of presentation, we define $\psi(t)$ and $\varphi(q) = \frac{k\phi_s^+(q) - \alpha_2(q)}{\alpha_1(q)}$.

According to Lemma \ref{lem:revenue}, the revenue function of the mediator for any feasible mechanism $(\pi, P_b, P_s)$ can be rewritten as:
\begin{align*}
    Rev(\pi,P_b,P_s) 
    = \int_{q_1}^{q_2} \int_{t_1}^{t_2}\pi(t, q)f(t)g(q)\alpha_1(q)[\psi(t) - \varphi(q) ]\,\mathrm{d}t\mathrm{d}q.
\end{align*}

Now we suppose that both function $\psi(t)$ and $\varphi(q)$ are not monotone increasing. We apply the so-called ironing technique to obtain the optimal mechanism. We first present the standard ironing procedure.
\begin{definition}[Ironing]
    Let $\psi(t)$ be any non-monotone function. 
\begin{enumerate}
	\item Define a random variable $w = F(t)$ and let
	\begin{align*}
		h_b(w) = \psi(F^{-1}(w)),
	\end{align*}
	where $F^{-1}(w)$ is the inverse function of $F(t)$. Note that $F(t)$ is continuous and strictly increasing since we assume that the density function $f(t)$ is always strictly positive. Thus the inverse function $F^{-1}$ is also continuous and increasing.
	\item Let $H_b:[0, 1] \mapsto \mathbb{R}$ be the integral of $h_b(w)$:
	\begin{align*}
		H_b(w) = \int_0^w h_b(r) \,\mathrm{d}r.
	\end{align*}  
	\item Let $L_b:[0, 1]\mapsto \mathbb{R}$ be the convex hull of the function $H_b$:
	\begin{align*}
		L_b(w) = \min \{ \lambda H_b(w_1) + (1-\lambda)H_b(w_2) \},
	\end{align*}
	where $\lambda, w_1, w_2 \in [0, 1]$ and $\lambda w_1 + (1-\lambda)w_2 =w$.
	\item Define $l_b:[0, 1] \mapsto \mathbb{R}$ such that
	\begin{align*}
		l_b(w) = L_b'(w).
	\end{align*}
	\item Then we obtain the ironed function $\bar{\psi}$:
	\begin{align*}
		\bar{\psi}(t) = l_b(w) = l_b(F(t)).
	\end{align*}
\end{enumerate}
\end{definition}

By definition, we have $h_b(F(t)) = \psi(t), l_b(F(t)) = \bar{\psi}(t)$. The mediator's revenue function can be rewritten as
\begin{align*}
	&\int_{q_1}^{q_2} \int_{t_1}^{t_2}\pi(t, q)f(t)g(q)\alpha_1(q)[\psi(t) - \varphi(q) ]\,\mathrm{d}t\mathrm{d}q\\
	=& \int_{q_1}^{q_2} \int_{t_1}^{t_2}\pi(t, q)f(t)g(q)\alpha_1(q)[\bar{\psi}(t) - \varphi(q) ]\,\mathrm{d}t\mathrm{d}q\\
	+&\int_{q_1}^{q_2} \int_{t_1}^{t_2}\pi(t, q)f(t)g(q)\alpha_1(q)[h_b(F(t)) - l_b(F(t))]\,\mathrm{d}t\mathrm{d}q
\end{align*}
We can simplify the second term above via integration by parts:
\begin{align*}
	&\int_{q_1}^{q_2} \int_{t_1}^{t_2}\pi(t, q)f(t)g(q)\alpha_1(q)[h_b(F(t)) - l_b(F(t))]\,\mathrm{d}t\mathrm{d}q\\
	=& \int_{t_1}^{t_2} [h_b(F(t)) - l_b(F(t))] R_b^{\pi}(t) \,\mathrm{d}F(t)\\
	=&[H_b(F(t)) - L_b(F(t))] R_b^{\pi}(t)|_{t_1}^{t_2} - \int_{t_1}^{t_2} [H_b(F(t)) - L_b(F(t))] \,\mathrm{d} R_b^{\pi}(t).
\end{align*}

By definition, $L_b$ is the convex hull of $H_b$, then we have $L_b(0)=H_b(0)$ and $L_b(1) = H_b(1)$. Therefore, the first term is 0, and the revenue function becomes:
\begin{align}
\begin{aligned}
\label{eq:iron buyer}
	Rev(\pi, P_b, P_s) =& \int_{q_1}^{q_2} \int_{t_1}^{t_2}\pi(t, q)f(t)g(q)\alpha_1(q)[\bar{\psi}(t) - \varphi(q) ]\,\mathrm{d}t\mathrm{d}q\\
	&-\int_{t_1}^{t_2} [H_b(F(t)) - L_b(F(t))] \,\mathrm{d} R_b^{\pi}(t).
\end{aligned}
\end{align}

To deal with function $\varphi(q)$, we modify the standard ironing by changing the definition of $w$ to the following:
\begin{gather*}
    w(q) = \int_{q_1}^q\alpha_1(r)g(r) \,\mathrm{d}r.
\end{gather*}
Note that $\alpha_1(q)>0$ for all $q\in [q_1, q_2]$ and that $\alpha_1(q)>0$. This means $w(q)$ is a strictly increasing function, and thus has an inverse function $w^{-1}:[0, \bar{w}] \mapsto [q_1, q_2]$, where $\bar{w} = \int_{q_1}^{q_2}\alpha_1(r)g(r) \,\mathrm{d}r$.

Let $h_s(w)$, $H_s(w)$, $l_s(w)$ and $L_s(w)$ be the corresponding functions in the ironing procedure when ironing function $\varphi(q)$. Define $\bar{\varphi}(q)=l_s(w(q))$.

By definition, we have $h_s(w(q)) = \varphi(q)$ and $l_s(w(q)) = \bar{\varphi}(q)$. Similar to the case of ironing the function $\psi(t)$, the first term of equation \eqref{eq:iron buyer} can be written as:
\begin{align*}
	&\int_{q_1}^{q_2} \int_{t_1}^{t_2}\pi(t, q)f(t)g(q)\alpha_1(q)[\bar{\psi}(t) - \varphi(q) ]\,\mathrm{d}t\mathrm{d}q\\
	=&\int_{q_1}^{q_2} \int_{t_1}^{t_2}\pi(t, q)f(t)g(q)\alpha_1(q)[\bar{\psi}(t) - \bar{\varphi}(q) ]\,\mathrm{d}t\mathrm{d}q\\
	&+\int_{q_1}^{q_2} \int_{t_1}^{t_2}\pi(t, q)f(t)g(q)\alpha_1(q)[l_s(w(q)) - h_s(w(q)) ]\,\mathrm{d}t\mathrm{d}q
\end{align*}
We can simplify the second term of the right-hand side as follows:
\begin{align*}
	&\int_{q_1}^{q_2} \int_{t_1}^{t_2}\pi(t, q)f(t)g(q)\alpha_1(q)[l_s(w(q)) - h_s(w(q)) ]\,\mathrm{d}t\mathrm{d}q\\
	=&\int_{0}^{\bar{w}} \int_{t_1}^{t_2}\pi(t, w^{-1}(\beta))f(t)[l_s(\beta) - h_s(\beta) ]\,\mathrm{d}t\mathrm{d}w\\
	=&\int_{q_1}^{q_2}[l_s(w(q)) - h_s(w(q))] R_s^{\pi}(q) \,\mathrm{d}w(q) \\
	=&[L_s(w(q)) - H_s(w(q))] R_s^{\pi}(q)|_{q_1}^{q_2} - \int_{q_1}^{q_2}[L_s(w(q)) - H_s(w(q))] \,\mathrm{d} R_s^{\pi}(q)
\end{align*}
$L_s(0) = H_s(0)$ and $L_s(\bar{w}) = H_s(\bar{w})$ since $L_s$ is the convex hull of $H_s$. Thus the first term is equal to 0. Overall, the revenue function of the mediator can be written as
\begin{align}
\begin{aligned}
\label{eq:iron rev}
	Rev(\pi,P_b, P_s) =& \int_{q_1}^{q_2} \int_{t_1}^{t_2}\pi(t, q)f(t)g(q)\alpha_1(q)[\bar{\psi}(t) - \bar{\varphi}(q) ]\,\mathrm{d}t\mathrm{d}q\\
	&-\int_{t_1}^{t_2} [H_b(F(t)) - L_b(F(t))] \,\mathrm{d} R_b^{\pi}(t)\\
	&-\int_{q_1}^{q_2}[L_s(W(q)) - H_s(W(q))] \,\mathrm{d} R_s^{\pi}(q).
\end{aligned}	
\end{align}

Now we present the optimal mechanism for this case, which also has a threshold structure.

\begin{theorem}
For any irregular case, the threshold mechanism with the following threshold functions and payment functions is an optimal mechanism:
\begin{gather*}
    \lambda(t) = \bar{\psi}(t)\quad \text{and}\quad \eta(q) = \bar{\varphi}(q),
\end{gather*}
\begin{align*}
    P_b^*(t) = &\frac{1}{\int_{q\in Q} \pi^*(t, q)g(q)\,\mathrm{d}q} \left[ \int_{q\in Q}\pi^*(t, q) v(t,q) g(q) \,\mathrm{d}q \right. \\
    &\left. - \int_{t_1}^t \int_{q\in Q}\alpha_1(q)\pi^*(x, q)g(q) \,\mathrm{d}q \mathrm{d}x \right],\\
    P_s^*(q) = &\frac{1}{\int_{t\in T}\pi^*(t, q)f(t) \,\mathrm{d}t} \left[\int_{t\in T}\pi^*(t, q) r(q) f(t) \,\mathrm{d}t \right.\\
    &\left. + k\int_{q}^{q_2} \int_{t\in T}\pi^*(t, x)f(t) \,\mathrm{d}t \mathrm{d}x \right]
\end{align*}
\label{thm:optimal_mechanism_irregular}
\end{theorem}

\begin{proof}

To prove the optimality of the above mechanism, we show that the mechanism $(\pi^*, P_b^*,P_s^*)$ maximizes all the terms in Equation \eqref{eq:iron rev} simultaneously. The first term is maximized by $\pi^*$, as $\pi^*(t, q) = 1$ only when $\bar{\psi}(t) - \bar{\varphi}(q) \ge 0$. Moreover, $(\pi^*, P_b^*, P_s^*)$ satisfies $U_b(t_1) = 0$ and $U_s(q_2) = 0$, indicating that it also maximizes the last two terms since Lemma \ref{lem:property} implies that $U_b(t_1)\ge 0$ and $U_s(q_2) \ge 0$ hold for all feasible mechanisms. 

As for the second term, it is worth noting that $H_b(F(t)) - L_b(F(t)) \ge 0$ since $L_b(w)$ is the convex hull of $H_b(w)$. Additionally, based on Lemma \ref{lem:property}, we know that $\mathrm{d} R_b^{\pi}(t) \ge 0$, which means this term is always non-negative. Therefore, in order to show that the second term is maximized, it suffices to prove that this term equals to 0. Actually, it is only necessary to consider cases where $H_b(F(t)) - L_b(F(t)) > 0$. In such cases, $t$ must fall within an ironed interval $I$, thus function $L_b(w)$ is linear in the interval $I$. This implies $l_b(w) = \bar{\psi}(t)$ is a constant and thus $R_b^{\pi}(t) = \int_{q: \bar{\varphi}(q)\le \bar{\psi}(t)} \alpha_1(q) g(q) \,\mathrm{d}q$ is also a constant in the interval $I$, leading to $\mathrm{d}R_b^{\pi^*}(t) = 0$.

Similarly, for the third term, we observe that $L_s(w(q)) - H_s(w(q)) \le 0$ since $L_s(w)$ is a convex hull of $H_s(w)$. According to Lemma \ref{lem:property}, $\mathrm{d} R_s^{\pi}(q) \le 0$ holds for all feasible mechanisms. Consequently, this term is also non-negative. We show that in the optima mechanism $(\pi^*,P_b^*, P_s^*)$, this term equals to 0. Consider cases where $L_s(w(q)) - H_s(w(q)) <0$. In such cases, $q$ must lie in an ironed interval $I$ and thus $L_s(w)$ is linear in the ironing interval. This implies $l_s(w) = \bar{\varphi}(q)$ is a constant. Then $R_s^{\pi^*}(q) = \int_{t:\bar{\psi}(t) \ge \bar{\varphi}(q)} f(t) \,\mathrm{d}t$ is also a constant in the interval $I$, leading to $\mathrm{d} R_s^{\pi^*}(q) = 0$.

In conclusion, the mechanism $(\pi^*, P_b^*, P_s^*)$ optimizes all terms in the Equation \eqref{eq:iron rev} simultaneously, proving it to be an optimal feasible mechanism.
\end{proof}

\section{Discussion}
Our main result reveals some interesting facts about the optimal mechanism.
\begin{fact}
	For a regular problem instance, the higher the seller's type is, the less likely their item will be sold.
\end{fact}
For a regular instance, the probability of $\lambda(t)\ge \eta(q)$ becomes smaller as $\eta(q)$ is monotone non-decreasing in $q$. Sellers with high-quality items also have higher reserve prices, which means that fewer buyers can afford to buy the item. 

\begin{fact}
	In the optimal mechanism, for each buyer of type $t$, there is a certain threshold $q'$, such that the buyer can trade with the seller only if the seller's type $q$ satisfies $q\le q'$.
	\label{fact:low_q}
\end{fact}
Fact \ref{fact:low_q} shows that even a buyer with a high type may still buy a low-quality item, but it is impossible for a low-typed buyer to buy a high-quality item.

\begin{fact}
	In the optimal mechanism, the mediator may lose money for some type profiles $(t, q)$.
\end{fact}
This happens whenever the mediator's payment to the seller is higher than the payment collected from the buyer, i.e., $P^*_s(q)>P^*_b(t)$. Let us re-consider Example \ref{example:1}. As shown in Figure \ref{fig:lose money}, the mediator will recommend trade in the area above the blue curve. And the equation $P^*_s(q)=P^*_b(t)$ is the orange curve. Therefore, in the area between the blue curve and the orange curve, the mediator recommends trade but loses money. 

\begin{figure}[ht]
\centering
\includegraphics[width=0.95\linewidth]{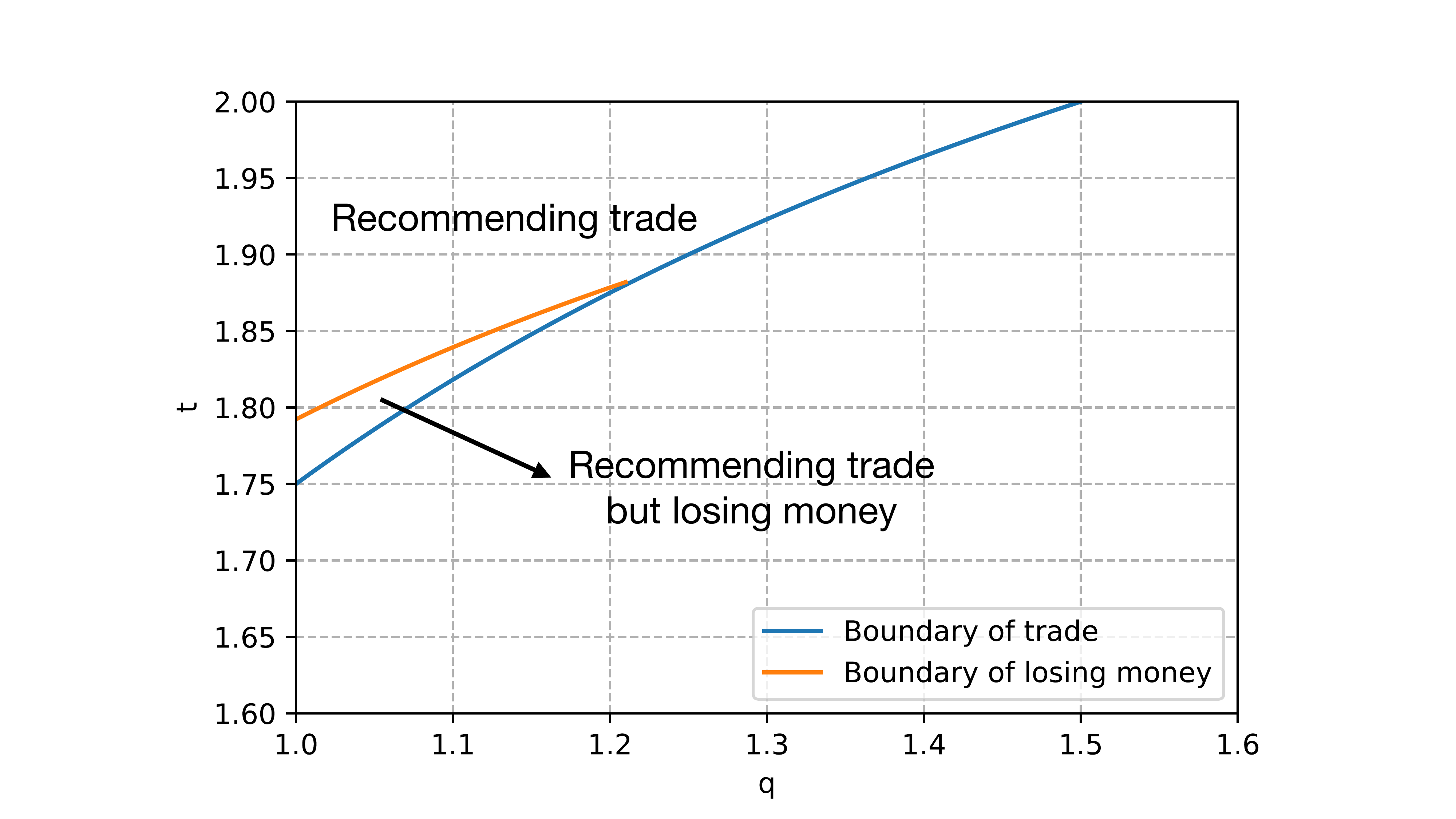}
\caption{The threshold mechanism in Example \ref{example:1}. } 
\label{fig:lose money}
\end{figure}

\section{Conclusion}
We studied the problem of designing revenue-maximizing mechanisms for the mediator. We formally defined the set of mechanisms that can be used by the mediator, and proved that the mediator can, without loss of generality, consider the set of direct and incentive compatible mechanisms. After that, we formulated this problem as a mathematical program, and gave an optimal solution in close-form under the regularity condition. We found that our optimal mechanism belongs to a category of \emph{threshold mechanisms}. Additionally, we extended our results to general cases by utilizing a variant version of the ironing technique proposed by~\cite{myerson1981optimal}. In the end, we discussed some interesting properties of the optimal mechanism.

\bibliographystyle{apalike}
\bibliography{bib}




\clearpage
\appendix
\section*{Appendix}
\section{Omitted Proofs in Section \ref{sec:space}}
\subsection{Proof of Theorem \ref{theorem:revelation}}

\begin{proof}
We construct a direct mechanism that simulates the original mechanism, thus giving the agents the same expected utilities and the mediator the same revenue. We can, without loss of generality, assume that the agents are always willing to trade if the mediator recommends ``trade'', since otherwise, the mediator can simply re-label the ``trade'' recommendation in this case as ``not trade'' without affecting the final outcome.

Let $M$ be any general mechanism and $\Gamma$ the induced game. Let $(\beta^*, \mu^*)$ be any PBE of $\Gamma$. Let $Z(t, q, \sigma)$ be the set of terminal nodes $z$ where the mediator recommends signal $\sigma$ along the path from the root to $z$ when the type profile is $(t,q)$. 
We construct the following direct mechanism:
\begin{enumerate}
    \item The mediator asks the buyer and the seller to report their types $t$ and  $q$;
    \item Let $\Sigma = \{0,1\}$ be a message set, and set $\pi(1 |t, q) =  \sum_{z\in Z(t,q,1)}\rho(z|\beta^*)$ to be the probability of recommending ``trade'';
    \item When the mediator recommends ``trade'', they charge the buyer $P_b(t)$ and pay the seller $ P_s(q)$, where
    \begin{gather*}
        P_b(t) = \frac{\E_{q \sim g}[\sum_{z\in Z(t,q)}\rho(z|\beta^*)\tau_b(z)]}{\E_{q \sim g}[\pi(1 |t, q)] },  \\
        P_s(q) = \frac{\E_{t \sim f} [\sum_{z\in Z(t,q)}\rho(z|\beta^*)\tau_s(z)] }{E_{t \sim f}[\pi(1 |t, q)] },
    \end{gather*}
    where $Z(t,q)=Z(t, q,0)\cup Z(t,q,1)$.
\end{enumerate}

We show that the new mechanism is indeed equivalent to the original one. In the new mechanism, the strategies of the buyer and the seller is simply to choose a type to report to the mediator. 

We start by checking that as long as the buyer and seller participate in the mechanism and report their type truthfully, their expected utilities are exactly the same in both two mechanisms. For simplicity, we only prove the statement for the buyer, as the argument for the seller is essentially the same. In the original mechanism, we have that $\sum_{z\in Z(t,q)}\rho(z|\beta^*)=1$. The expected utility of the buyer is:
\begin{align*}
    &\E_{q\sim g}\left[\sum_{z\in Z(t,q)}\rho(z|\beta^*)u_b(z)\right]\\
    =&\E_{q\sim g}\left[\sum_{z\in Z(t, q,0)}\rho(z|\beta^*)u_b(z)+\sum_{z\in Z(t, q,1)}\rho(z|\beta^*)u_b(z)\right]\\
    =&\E_{q\sim g}\left[\sum_{z\in Z(t, q,1)}\rho(z|\beta^*)v(q,t)-\sum_{z\in Z(t,q)}\rho(z|\beta^*)\tau_b(z)\right],
\end{align*}
where the last equation is obtained by plugging in the buyer's utility function.

In the new mechanism, the mediator's signal set only contains two signals, 1 for ``trade'' and 0 for ``not trade''. Let $g(q|t,1)$ be the buyer's belief of the seller's type upon receiving signal 1. We have
\begin{align*}
    g(q|t,1)=\frac{\pi(1|t,q)g(q|t)}{ \int_q\pi(1|t,q)g(q|t)\,\dd q}=\frac{\pi(1|t,q)g(q)}{ \int_q\pi(1|t,q)g(q)\,\dd q}.
\end{align*}
Therefore, the buyer's expected utility in this case is:
\begin{align*}
    &\E_{q\sim g(q|t,1)}\left[v(t,q)\right]-P_b(t)\\
    =&\frac{\int_q\pi(1|t,q)g(q)v(t,q)\,\dd q-\E_{q \sim g}\left[\sum_{z\in Z(t,q)}\rho(z|\beta^*)\tau_b(z)\right]}{\int_q\pi(1|t,q)g(q)\,\dd q}.
\end{align*}
When receiving signal 0, the buyer's expected utility is simply 0, as there is no trade and the mediator charges nothing. Therefore, the expected utility of a buyer with type $t$ in the new mechanism is:
\begin{align*}
    &\pi(1|t)\left[\E_{q\sim g(q|t,1)}\left[v(t,q)\right]-P_b(t)\right]\\
    =&\int_q\pi(1|t,q)g(q)v(t,q)\,\dd q-\E_{q \sim g}\left[\sum_{z\in Z(t,q)}\rho(z|\beta^*)\tau_b(z)\right]\\
    =&\E_{q\sim g}\left[\sum_{z\in Z(t,q,1)}\rho(z|\beta^*)v(q,t)-\sum_{z\in Z(t,q)}\rho(z|\beta^*)\tau_b(z)\right],
\end{align*}
where $\pi(1|t)$ is the probability that a buyer of type $t$ receives signal 1, and the last equation is due to the construction of $\pi(1|t,q)$. 

The above analysis shows that truthful reporting gives the buyer the same expected utility in the two mechanisms.

Now we prove that, for any agent, misreporting their type leads to the same utility as if they pretend to have the misreported type in the original mechanism. Let $\beta'_b$ be the equilibrium strategy of a buyer of type $t'$, and $\beta'=(\beta'_b,\beta^*_s)$ be the corresponding strategy profile. If a buyer of type $t$ uses $\beta'_b$ (pretends to have type $t'$), their expected utility in the original mechanism is:
\begin{align*}
    \E_{q\sim g}\left[\sum_{z\in Z(t',q,1)}\rho(z|\beta')v(q,t')-\sum_{z\in Z(t',q)}\rho(z|\beta')\tau_b(z)\right].
\end{align*}
One can verify that this is also the expected utility of the buyer reporting $t'$ in the new mechanism.

We can also argue similarly for the seller. It follows that if any player can profit from misreporting in the new mechanism, they could also obtain a higher utility in the original one by pretending to have the misreported type, which is impossible as $\beta^*$ is the equilibrium in the original mechanism. Therefore, the new mechanism is incentive compatible.

\end{proof}

\section{Omitted Proofs in Section \ref{sec:optimal}}
\subsection{Proof of Lemma \ref{lem:property}}

\begin{proof}
For simplicity, we only prove for the buyer but omit similar proof for the seller.

We first prove that the constraints are necessary for the buyer. Slightly manipulating the buyer's IC constraint \eqref{eq:IC buyer} yields:
\begin{align*}
    U_b(t) &\ge U_b(t') + \int_{q\in Q} \pi(t',q)[v(t,q) - v(t', q)] g(q) \,\mathrm{d}q \\
    &= U_b(t') + (t-t')\int_{q\in Q} \alpha_1(q) \pi(t',q) g(q) \,\mathrm{d}q.
\end{align*}

Note that $R_b^{\pi}(t) = \int_{q \in Q} \alpha_1(q) \pi(t, q)g(q)  \,\mathrm{d}q$, thus the above inequality is equivalent to:
\begin{gather}
\label{eq:buyer_ic_1}
    U_b(t) - U_b(t') \ge (t-t')R^{\pi}_b(t').
\end{gather}
The above equation should hold for any $t$ and $t'$. As a result, if we switch $t$ and $t'$, we should have:
\begin{gather}
\label{eq:buyer_ic_2}
    U_b(t') - U_b(t) \ge (t'-t)R^{\pi}_b(t).
\end{gather}
Combining Equation \eqref{eq:buyer_ic_1} and \eqref{eq:buyer_ic_2} gives:
\begin{gather}
    (t-t') R^{\pi}_b(t') \le U_b(t) - U_b(t') \le (t-t')R^{\pi}_b(t).
\end{gather}
When $t>t'$, we can divide the above inequality by $t-t'$ and get:
\begin{gather}
    R^{\pi}_b(t') \le \frac{U_b(t) - U_b(t')}{t - t'} \le R^{\pi}_b(t),
\end{gather}
Letting $t \rightarrow t'$, the above inequalities become:
\begin{gather}
    \frac{\mathrm{d} U_b(t)}{\mathrm{d}t} = R^{\pi}_b(t).
    \label{eq:utility_derivative}
\end{gather}
The above equation still holds if $t<t'$. Therefore, Equation \eqref{eq:IC property buyer} follows.

Adding Equation \eqref{eq:buyer_ic_1} and \eqref{eq:buyer_ic_2} gives $(t-t')[R_b^{\pi}(t)-R_b^{\pi}(t')]\ge 0$,
which implies constraint \eqref{eq:buyer_monotone}.

Now let's consider the IR constraint for the buyer. By definition, $R^{\pi}_b(t)$ is non-negative. Thus, with Equation \eqref{eq:IC property buyer}, we know that $U_b(t)$ attains its minimum value at $t_1$. Therefore, to satisfy the IR constraint, we only need to ensure $U_b(t_1)\ge 0$, which is exactly Equation \eqref{eq:IR property buyer}.

We omit the proof for constraint \eqref{eq:seller_monotone}, Equation \eqref{eq:IC property seller}, and Equation \eqref{eq:IR property seller}, as they can be obtained through similar analysis for the seller.

Now we show that the constraints are also sufficient for the buyer. The IC constraint \eqref{eq:IC buyer} for the buyer is equivalent to:
\begin{gather*}
    U_b(t) \ge U_b(t') + (t-t')\int_{q \in Q} \alpha_1(q) \pi(t',q)g(q) \,\mathrm{d}q.
\end{gather*}
Constraint \eqref{eq:IC property buyer} implies the above inequality, because if $t' < t$, we have
\begin{align*}
    U_b(t)-U_b(t') = \int_{t'}^t R^{\pi}_b(x)\,\mathrm{d}x \ge (t -t') R^{\pi}_b(t').
\end{align*}
Similarly, when $t' > t$, we also have $U_b(t) - U_b(t') \ge  (t -t') R^{\pi}_b(t') $.

The IR constraint \eqref{eq:IR buyer} is equivalent to $U_b(t) \ge 0$. $R^{\pi}_b(t') \ge 0$ and $R^{\pi}_b(t)$ is monotone non-decreasing implies $R^{\pi}_b(t) \ge 0, \forall t$. Therefore, Equation \eqref{eq:IC property buyer} and \eqref{eq:IR property buyer} together shows that $U_b(t) \ge 0$ for any $t$.
\end{proof}

\subsection{Proof of Lemma \ref{lem:revenue}}
\begin{proof}
For any feasible mechanism $(\pi, P_b, P_s)$, the revenue of the mediator can be written as:
\begin{align}
\label{eq:revenue}
\begin{aligned}
    Rev(\pi,P_b,P_s) = &\int_{t\in T} f(t) \int_{q\in Q}g(q)\pi(t, q) P_b(t) \,\mathrm{d}q\mathrm{d}t -\\
    &-\int_{q\in Q}g(q) \int_{t \in T} f(t) \pi(t, q) P_s(q) \,\mathrm{d}t\mathrm{d}q .
\end{aligned}
\end{align}

Using Equation \eqref{eq:utility buyer}, the first term of the above equation can be written as:
\begin{align*}
    &\int_{t_1}^{t_2} f(t) \int_{q\in Q}g(q)\pi(t, q) P_b(t)\,\mathrm{d}q\mathrm{d}t \\
    =& \int_{t_1}^{t_2} f(t) \left[\int_{q\in Q}\pi(t, q) v(t,q) g(q)\,\mathrm{d}q - U_b(t)  \right] \,\mathrm{d}t, 
\end{align*}
where the second term containing $U_b(t)$ can be further expanded as follows:
\begin{align*}
    & \int_{t_1}^{t_2} f(t) U_b(t) \,\mathrm{d}t\\
    =& \int_{t_1}^{t_2} f(t) \left[ \int_{t_1}^t R^{\pi}_b(x)\,\mathrm{d}x + U_b(t_1)\right] \,\mathrm{d}t\\
    =&\int_{t_1}^{t_2}\int_{t_1}^t f(t)R^{\pi}_b(x)\,\mathrm{d}x\mathrm{d}t + U_b(t_1)\\
    =& \int_{t_1}^{t_2}\int_{x}^{t_2} f(t)R^{\pi}_b(x)\,\mathrm{d}t\mathrm{d}x + U_b(t_1)\\
    = &\int_{t_1}^{t_2}[1-F(x)]R^{\pi}_b(x)\,\mathrm{d}x + U_b(t_1)\\
    =& \int_{q\in Q} g(q)\left[\int_{t_1}^{t_2}[1-F(t)]\alpha_1(q) \pi(t, q) \,\mathrm{d}t \right]\,\mathrm{d}q + U_b(t_1).
\end{align*}
Plugging the term back into the above equation, and switching the order of integral gives:
\begin{align*}
    &\int_{t_1}^{t_2} f(t) \int_{q\in Q}g(q)\pi(t, q) P_b(t)\,\mathrm{d}q\mathrm{d}t\\
    =&\int_{q\in Q} g(q) \left[\int_{t_1}^{t_2} f(t)\pi(t, q)[\alpha_1(q) \phi^-_b(t)+\alpha_2(q) ] \,\mathrm{d}t \right]\,\mathrm{d}q \\
    &- U_b(t_1).
\end{align*}

The second term of Equation \eqref{eq:revenue} is about the seller. According to Equation \eqref{eq:utility seller}, it can be written as:
\begin{align}
    &\int_{q_1}^{q_2} g(q) \int_{t\in T} f(t) \pi(t, q)P_s(q) \,\mathrm{d}t\mathrm{d}q \nonumber\\
    =&\int_{q_1}^{q_2} g(q) \left[\int_{t\in T}\pi(t, q)r(q)f(t)\,\mathrm{d}t + SU(q) \right] \,\mathrm{d}q,\label{eq:rev_seller}
\end{align}
where we can re-write the second term as:
\begin{align*}
    &\int_{q_1}^{q_2} g(q) SU(q) \,\mathrm{d} q\\
    =& \int_{q_1}^{q_2} g(q) \left[k\int_q^{q_2} R^{\pi}_s(x)\,\mathrm{d}x + SU(q_2) \right]\,\mathrm{d}q\\
    =&k\int_{q_1}^{q_2} g(q)\int_{q}^{q_2}R^{\pi}_s(x)\,\mathrm{d}x\mathrm{d}q + SU(q_2)\\
    =&k\int_{q_1}^{q_2}\int_{q_1}^{x}g(q)R^{\pi}_s(x)\,\mathrm{d}q\mathrm{d}x + SU(q_2)\\
    =&k\int_{q_1}^{q_2}G(x)R^{\pi}_s(x)\,\mathrm{d}x + SU(q_2)\\
    =&k\int_{t}f(t)\int_{q_1}^{q_2}G(q)\pi(t, q)\,\mathrm{d}q \mathrm{d}t + SU(q_2)
\end{align*}
Putting everything back to Equation \eqref{eq:rev_seller} yields:
\begin{align*}
    &\int_{q_1}^{q_2} g(q) \int_{t\in T} f(t) \pi(t, q)P_s(q) \,\mathrm{d}t\mathrm{d}q \\
    = &\int_{t \in T} f(t) \left[\int_{q_1}^{q_2}g(q)\pi(t, q) k\left(q +\frac{G(q)}{g(q)}\right)\,\mathrm{d}q\right]\,\mathrm{d}t + SU(q_2)\\
    =&\int_{t \in T} f(t) \left[\int_{q_1}^{q_2}g(q)\pi(t, q) k\phi^+_s(q)\,\mathrm{d}q\right]\,\mathrm{d}t + SU(q_2),
\end{align*}
By combining both two terms of Equation \eqref{eq:revenue}, we get:
\begin{align*}
    & Rev(\pi,P_b,P_s) \\
    = & \int_{q_1}^{q_2} \int_{t_1}^{t_2}\pi(t, q)f(t)g(q)[\alpha_1(q) \phi^-_b(t)  \\
    &+\alpha_2(q)- k \phi^+_s(q)]\,\mathrm{d}t\mathrm{d}q-U_b(t_1)-SU(q_2).
\end{align*}
\end{proof}

\subsection{Proof of Lemma \ref{lem:feasible}}
\begin{proof}
Observe that if a problem instance is regular, the recommendation scheme $\pi^*(t,q)$ is non-decreasing in $t$ but non-increasing in $q$. This implies that $R^\pi_b(t)$ is non-decreasing in $t$ and $R^\pi_s(q)$ is non-increasing in $q$. Satisfying constraint \eqref{eq:buyer_monotone} and \eqref{eq:seller_monotone}.

By definition, the buyer's utility is:
\begin{align*}
    U_b(t) =& \int_{q\in Q} \pi^*(t, q)[v(t,q) - P^*_b(t)] g(q)\,\mathrm{d}q \\
    =&\int_{t_1}^t \int_{q\in Q}\alpha_1(q)\pi^*(x, q)g(q) \,\mathrm{d}q \mathrm{d}x.
\end{align*}
This means $U_b(t_1)=0$ and that
\begin{align*}
	\frac{\dd U_b(t)}{\dd t}=\int_{q\in Q}\alpha_1(q)\pi^*(t, q)g(q) \,\mathrm{d}q=R_b^{\pi^*}(t).
\end{align*}
Therefore, Equation \eqref{eq:IC property buyer} follows.

Similarly, by definition, the seller's surplus is:
\begin{align*}
    SU(q) =& \int_{t\in T} \pi^*(t, q)[P_s^*(q) - r(q)]f(t)\,\mathrm{d}t \\
    =& k\int_{q}^{q_2} \int_{t\in T}\pi^*(t, x)f(t) \,\mathrm{d}t \mathrm{d}x.
\end{align*}
It follows that $SU(q_2)=0$ and that
\begin{align*}
	\frac{\dd SU(q)}{\dd q}=-k\int_{t\in T}\pi^*(t, q)f(t) \,\mathrm{d}t=-kR_s^{\pi^*}(q).
\end{align*}
These clearly imply Equation \eqref{eq:IC property seller}.
\end{proof}

\subsection{Proof of Theorem \ref{thm:optimal_mechanism}}
\begin{proof}[Proof of Theorem \ref{thm:optimal_mechanism}]
The revenue equation \eqref{eq:prove revenue} contains 3 terms. We now show that the mechanism $(\pi^*, P_b^*, P_s^*)$ optimizes them all at the same time. 

 By definition, $\pi^*(t,q)=1$ if and only if $\lambda(t)\ge \eta(q)$, which is equivalent to $\alpha_1(q) \phi^-_b(t)+\alpha_2(q)- k \phi^+_s(q)\ge 0$. This means that the first term is point-wisely optimized for all $(t,q)$ pairs, and thus the term is maximized.

For the second and third terms, if we want to maximize the revenue, we need to minimize both these two terms. According to Lemma \ref{lem:property}, both $U_b(t_1)$ and $SU(q_2)$ need to be non-negative. The proof of Lemma \ref{lem:feasible} shows that the mechanism $(\pi^*, P_b^*, P_s^*)$ satisfies $U_b(t_1)=0$ and $SU(q_2)=0$, which implies that both $U_b(t_1)$ and $SU(q_2)$ are already optimized.
\end{proof}

\subsection{Proof of Lemma \ref{lem: payment}}
\begin{proof}
We prove it by contradiction. According to Lemma \ref{lem:revenue}, given a feasible mechanism $(\pi, P_b, P_s)$, the mediator's expected revenue is:
    \begin{align}
    & Rev(\pi,P_b,P_s)\nonumber\\
    = & \int_{q_1}^{q_2} \int_{t_1}^{t_2}\pi(t, q)f(t)g(q)[\alpha_1(q) \phi^-_b(t)  \nonumber\\
    &+\alpha_2(q)- k \phi^+_s(q)]\,\mathrm{d}t\mathrm{d}q-U_b(t_1)-U_s(q_2).\label{eq:prove revenue}
    \end{align}

If $(\pi, P_b, P_s)$ is a revenue-maximizing mechanism and $U_b(t_1) >0$. We construct a new mechanism $(\pi, P'_b, P_s)$, in other words, a new payment function of the buyer. The way of construction is as follows:
    \begin{gather}\label{eq:constructed mechanism}
        P'_b(t) =
        \begin{cases}
            P_b(t) & \text{if $t \neq t_1$}\\
            P_b(t) + c & \text{otherwise}
        \end{cases},
    \end{gather}
where $c = \frac{1}{\int_{q\in Q} \pi(t_1, q)g(q) \,\mathrm{d}q } \int_{q\in Q} \pi(t_1, q) [v(t_1, q) - P_b(t_1)]$. We first show that the constructed mechanism $(\pi, P'_b, P_s)$ is also a feasible mechanism. Note that the monotone constraints \eqref{eq:buyer_monotone} and \eqref{eq:seller_monotone} only depend on $\pi$, and $\pi$ is unchanged in the constructed mechanism, thus the constructed mechanism satisfies these two constraints. Moreover, the utility of the seller in $q_2$ is also unchanged since $P_s$ is unchanged. So the constructed mechanism also satisfies constraint \eqref{eq:IR property seller}. Combining \eqref{eq:constructed mechanism} and \eqref{eq:utility buyer}, the utility of the buyer in $t_1$ is equal to $0$, i.e., $U_b(t_1) = 0$, which satisfies constraint \eqref{eq:IR property buyer}. Therefore, the constructed mechanism $(\pi, P'_b, P_s)$ is a feasible mechanism.

Then we show that the revenue of the constructed mechanism is greater than the original one. Note that the first term and last term in \eqref{lem:revenue} are unchanged since $\pi$ and $P_s$ are unchanged. And the second term in \eqref{lem:revenue} under the constructed mechanism is equal to $0$, which is less than that of the original mechanism. Thus the revenue of the constructed mechanism is greater than the original one, which contradicts the fact that $(\pi, P_b, P_s)$ is a revenue-maximizing mechanism. Therefore, a feasible mechanism $(\pi, P_b, P_s)$ is a revenue-maximizing mechanism, it must satisfy $U_b(t_1) = 0$. 
\end{proof}

\end{document}